\documentclass[10pt]{iopart}

\usepackage[all]{xy} 
\usepackage{iopams}
\usepackage{amsthm}
\usepackage{graphicx}

\usepackage{bm}
\usepackage{dsfont}
\usepackage{mathrsfs}
\usepackage[T1]{fontenc}
\usepackage{la}

\usepackage[usenames]{color}
\usepackage{hyperref}

\newtheorem{theorem}{Theorem}
\newtheorem{lemma}{Lemma}

\newtheorem{definition}{Definition}

\newcommand{\theref}[1]{theorem~\ref{#1}} 
\newcommand{\lemref}[1]{lemma~\ref{#1}} 
\newcommand{\defref}[1]{definition~\ref{#1}}

\newcommand\RR{{\mathds{R}}}
\newcommand\CC{{\mathds{C}}}

\newcommand{\prs}{{\bm\sigma}}
\newcommand{\abs}[1]{\left| #1\right|}
\newcommand{\avg}[1]{\left\langle #1 \right\rangle}
\newcommand{\adj}{\mathrm{adj}}
\newcommand\one{{\mathds{1}}}
\newcommand\eps{{\varepsilon}}
\newcommand{\hilb}{{\mathcal S}}
\newcommand{\Hilb}{{\mathcal F}}
\newcommand{\setP}{\mathfrak{P}}
\newcommand{\setI}{\mathcal{I}}
\newcommand{\nodes}{\mathfrak{M}}
\newcommand{\modes}{\mathfrak{N}}
\newcommand{\Kodes}{\mathfrak{O}}
\newcommand{\kodes}{\mathfrak{o}}
\newcommand{\edges}{\mathfrak{A}}
\renewcommand{\graph}{\mathcal{G}}
\newcommand{\liouv}{\mathscr{L}}
\newcommand{\aherm}{\mathscr{A}}

\newcommand{\green}{\mathscr{G}}
\newcommand{\oK}{\overline{K}}
\newcommand{\oW}{\overline{W}}
\newcommand{\oJ}{\overline{J}}
\newcommand{\trsp}{T} 
\renewcommand{\le}{\leqslant}
\renewcommand{\ge}{\geqslant}
\newcommand{\tpe}{t_{\perp}}
\newcommand{\tpi}{t_{\parallel}}
\newcommand{\pb}[1]{\left[#1\right]} 
\newcommand{\vv}[1]{(#1)} 

\begin{document}

\title{Symmetries of a mean-field spin model}

\author{Rytis Pa\v{s}kauskas}
\address{National Institute for Theoretical Physics, Stellenbosch 7600, South Africa}
\ead{\mailto{rytis.paskauskas@gmail.com}}

\begin{abstract}
Thermodynamic limit evolution of a closed quantum Heisenberg-type spin model with mean-field interactions is characterized by classifying all the symmetries of the equations of motion. It is shown that parameters of the model induce a structure in the Hilbert space by partitioning it into invariant subspaces, decoupled by the underlying Bogoliubov-Born-Green-Kirkwood-Yvon ({BBGKY}) hierarchy. All possible partitions are classified in terms of a $3\times3$ matrix of effective, thermodynamic limit coupling constants. It is found that there are either $1$, $2$, $4$, or $O(N)$ invariant subspaces. The {BBGKY} hierarchy decouples into the corresponding number of anti-Hermitian operators on each subspace. These findings imply that equilibration and the equilibrium in this model depend on the initial conditions.
\end{abstract}


\section{Introduction}\label{s:intro}
Recent developments in ultracold atom and ion experiments have spurred renewed interest in dynamics of closed quantum many-body systems \cite{BlochDalibardZwerger08,Kinoshita_etal06}. Although it is expected that equilibration (in a suitable sense) takes place in a generic non-integrable quantum system \cite{vonNeumann}, a substantial body of recent research has shown that the interim evolution, which is becoming  accessible experimentally, is not yet fully understood \cite{Polkovnikov_etal11}. In particular, several qualitatively different scenarios of equilibration that depend on the initial state, including relaxation to a state that is different from the thermodynamic equilibrium, have been reported in a closed non-integrable spin chain \cite{banuls2011}. In this article a different spin model is studied, where the equilibrium state may not be independent of the initial conditions.

A model, where $N$ spin-$1/2$ particles on a lattice interact by a mean-field Curie-Weiss potential, has been discussed in \cite{paskauskas2012a}. From a previous study of a related model \cite{Kastner10}, it had been known that, with a restricted choice of coupling constants and initial conditions, certain observables approach their microcanonical expectation values with a rate $t_0$ that is system size-dependent and, in fact, diverges as $t_0\sim N^{1/2}$. With the aim of extending the range of parameters, and better understanding the course of equilibration of homogeneous states with arbitrary initial conditions, a study of the general Curie-Weiss quantum Heisenberg model was undertaken in \cite{paskauskas2012a}. For that purpose, reduced density operators $F_{1\dots n}$ for all $n=1\dots N$ spins were expanded in terms of a set of tensorial coefficients $\{f_n\}_{n=0}^N$, and the evolution equations for all $f_n$ were derived, starting from the Bogoliubov-Born-Green-Kirkwood-Yvon ({BBGKY}) hierarchy. Remarkably, these evolution equations, expressed in scaled units of time $\tau=t/t_0$ have a non-trivial $N=\infty$ limit. This thermodynamic \emph{coefficient {BBGKY} hierarchy} implies that in thermodynamic limit, $t_0$ is the slowest time scale of equilibration for arbitrary initial conditions and parameters of the model.

However, scaling is not the entire story, and additional information about equilibration is contained in the solutions of the equations which, unlike scaling, depend on a $3\times3$ matrix of spin-spin coupling constants $J$ and on the magnetic field vector $h$ in a non-trivial way. For example, it was shown that if the Curie-Weiss potential is of the form $V=-NJ_\perp(S^xS^x+S^yS^y)-NJ_zS^zS^z$ [ which corresponds to a diagonal matrix $J=\mathrm{diag}(J_\perp,J_\perp, J_z)$], and where $S^a$ is the $a\in\{x,y,z\}$ component of the average spin per particle, the coefficient BBGKY equations simplify considerably. Evolution of a certain subset of tensorial coefficients $f_n$ could be determined analytically to all orders $n=1,2,\dots$, showing a characteristic superexponential [$\sim\exp{(-\tau^2\Delta/2)}$, where $\Delta=(J_z-J_\perp)^2$]  approach to thermodynamic expectation values, which in this case are equal to zero. And yet, another subset of coefficients could be identified, which displays no relaxation at all. These analytic results are not expected to be qualitatively valid for more general couplings and, in fact, no such analytic results are known about as simple generalization, as $J=\mathrm{diag}(J_x,J_y,J_z)$, which differs from the one considered in \cite{paskauskas2012a} by allowing all three diagonal elements of $J$ to be different. 

In this article, the thermodynamic coefficient {BBGKY} hierarchy equations are further investigated. The main goal is to classify, for a general $3\times3$ matrix $J$ of coupling constants and a general magnetic field vector $h=(h_x,h_y,h_z)$, all tensorial coefficients with different behavior. Although the number $D(N)$ of such coefficients depends on the system size $N$ and not on parameters $J$ or $h$, it will be shown that the matrix of effective thermodynamic limit coupling constants $\oJ(J,h)$, which is a function of $J$ and $h$, determines the partition of all the coefficients into several groups, decoupled by the {BBGKY} hierarchy. In this sense, $\oJ$ induces a topology in the vector space space of expansion coefficients.

This study provides an insight that the $N$-body quantum dynamics may impose symmetries on the Hilbert space. It is believed in this work that similar symmetries are present in a class of more realistic lattice spin models, where interactions among particles depend on the distance between them. Existence of invariant subspaces implies that equilibration and the equilibrium state may not be independent of the initial state, but not only that. It allows to divide a large computational problem into several smaller ones and moreover, should be taken in consideration when approximate dynamical models, such as kinetic theories, are constructed. The latter include a closure scheme as one step of approximation, whereby an ad-hoc relationship among correlators, involving different numbers of particles, is postulated. The existence of invariant subspaces with rigorously independent evolution implies causal relations among spin-spin correlators which should be respected by a good closure scheme.

The article is structured as follows. \Sref{s:summary} contains the preliminaries: The Curie-Weiss anisotropic quantum Heisenberg model is presented, a special expansion of reduced density operators is recalled, and the derivation of the equations for the expansion coefficients is outlined \cite{paskauskas2012a}.
In \sref{s:equations}, a useful representation of the coefficient {BBGKY} hierarchy equations is derived. Moreover, since the subsequent analysis of the equations can be succinctly expressed using concepts of graph theory, part of \sref{s:equations} is dedicated to the introduction to the relevant terminology.
In \sref{s:transformation}, equations are reformulated using the generalized Liouville operator. It is shown that this generalized Liouville operator is anti-Hermitian, provided that the scalar product in the vector space is appropriately defined.
In \sref{s:symmetries} the partition of the coefficient vector space is investigated. It is shown that all the coefficients can be grouped in either four, two or $O(N)$ disjoint sets, depending on whether the matrix of effective coupling constants is diagonal isotropic or anisotropic, or has a certain number of non-zero off-diagonal elements.
Finally, in \sref{s:conclusions} conclusions are drawn. 

\section{Model and a summary of previous results}\label{s:summary}
This section briefly outlines some of the previous results which will be necessary for the subsequent discussion. The Reader is kindly referred to \cite{paskauskas2012a} for details of derivation.

\subsection{Curie-Weiss anisotropic quantum Heisenberg model}
Consider a Hilbert space $\mathcal{H}_N=\bigotimes_{i=1}^N \CC_i^2$ of $N$ identical spin-$1/2$ particles, where the Hilbert space of each spin-$1/2$ particle is $\CC_i^2$. The anisotropic Curie-Weiss Hamiltonian $H_{1\dots N}$, acting on $\mathcal{H}_N$, is defined by
\begin{equation}
  H_{1\dots N} = \sum_{i=1}^N H_i + \sum_{\scriptstyle i,j=1\atop \scriptstyle i<j}^N V_{ij}\,.\label{e:ham}
\end{equation}
The local potential $H_i$ and the interaction potential $V_{ij}$ are defined as
\begin{equation}
  H_i   = -\sum_{a\in\setI} h_a \sigma_i^a\,,\qquad V_{ij} = -\lambda\sum_{a,b\in\setI} J_{ab}\sigma_i^a \sigma_j^b\,. \label{e:V}
\end{equation}
where $\setI=\{x,y,z\}$, and where
\begin{equation}\label{e:lambda}
  \lambda=1/N.
\end{equation}
The scaling factor $\lambda$ is necessary to render the energy per spin finite in the thermodynamic limit. The interaction potential $V_{ij}$ is determined by a symmetric $3\times3$ matrix of spin-spin coupling constants $J=(J_{ab})$. The local potential $H_i$ is determined by the three components of the magnetic field vector $h=(h_a)$.

\subsection{Operator {BBGKY} hierarchy}
The density operator of an $N$-spin system is a self-adjoint, positive, trace-class operator $\rho_N\in \mathcal{H}_N$. The normalization convention $\Tr_{1\dots N}\rho_N=1$ will be used, where $\Tr_{1\dots N}$ denotes the trace over $N$ degrees of freedom. The density operator satisfies the Von Neumann equation, associated to the Hamiltonian \eref{e:ham}
\begin{equation}\label{e:vn}
  \rmi\hbar\partial_t \rho_N=\pb{H_{1\ldots N},\rho_N}.
\end{equation}
Reduced $n$-particle density operators are defined by 
\begin{equation}\label{e:F}
  F_{1\dots n}=\Tr_{n+1\dots N}\rho_N \qquad  1\le n\le N.
\end{equation}
where $\Tr_{n+1\dots N}$ denotes a partial trace over $N-n$ spin degrees of freedom. In this article, only homogeneous states will be considered. The $\rho_N$ for such states is invariant with respect to an arbitrary permutation of particle indices. As a result, the $n$-particle reduced density operators do not depend on \emph{which} $N-n$ particle complex is traced over in \eref{e:F}. 

Reduced density operators satisfy \emph{trace properties} which are relations of the form
\begin{equation}\label{e:traceprop}
  \Tr_{n+1} F_{1\dots n+1} = F_{1\dots n}.
\end{equation}

The set of evolution equations for all the operators $\{F_{1\dots n}\}_{n=1}^N$ is called the Bogoliubov-Born-Green-Kirkwood-Yvon (BBGKY) hierarchy. It is derived from \eref{e:vn} using the standard formalism \cite{balescu1975}. For the Hamiltonian \eref{e:ham} it reads
\begin{equation}\label{e:spinkin}
  \rmi\hbar\partial_t F_{1\ldots n}=\sum_{\scriptstyle i,j=1\atop \scriptstyle i<j}^n\pb{V_{ij},F_{1\ldots n}}+ (N-n) \Tr_{n+1}\sum_{i=1}^{n}\pb{V_{i,n+1},F_{1\ldots n+1}}.
\end{equation}
The initial value problem for the reduced density operators $F_{1\dots n}$ consists of the set of equations \eref{e:spinkin} for each $1\le n\le N$, supplemented by the trace properties \eref{e:traceprop}.

It was shown \cite{paskauskas2012a} that there are two time scales involved in the dynamics: The slow time scale of long-range spin-spin interactions drives the equilibration and is the main focus of this article. This slow dynamics is coupled with the fast precession around the magnetic field vector $h=(h_x,h_y,h_z)$, appearing in the terms $H_i$. The latter terms can be eliminated by passing to an interaction picture, defined by
\begin{equation}\label{e:iPic}
  F_{1\ldots n}\rightarrow
  \tilde{F}_{1\dots n}= \exp\biggl[\frac{\rmi t}{\hbar}\sum_{i=1}^ n H_i\biggr] F_{1\ldots n} \exp\biggl[-\frac{\rmi t}{\hbar}\sum_{i=1}^ n H_i\biggr].
\end{equation}
The {BBGKY} hierarchy in the interaction picture is then given by
\begin{equation}\label{e:spinkin2}
  \fl \rmi\hbar\partial_t \tilde{F}_{1\ldots n} =
  \sum_{\scriptstyle i,j=1\atop \scriptstyle i<j}^n\pb{\tilde{V}_{ij}(t),\tilde{F}_{1\ldots n}}+ (N-n) \Tr_{n+1}\sum_{i=1}^{n}\pb{\tilde{V}_{i, n+1}(t),\tilde{F}_{1\ldots n+1}},
\end{equation}
where the two-body potential is now explicitly time-dependent,
\begin{equation}\label{e:V2}
  \tilde{V}_{ij}(t) =
  \exp\left[\frac{\rmi t}{\hbar} \left(H_i+H_j\right)\right]V_{ij}\exp\left[-\frac{\rmi t}{\hbar} \left(H_i+H_j\right)\right].
\end{equation}
The interaction picture potential $\tilde{V}_{ij}(t)$ is of the same form as \eref{e:V}, 
\begin{equation}
  \tilde{V}_{ij}(t)=-\lambda \sum_{a,b\in\setI} \tilde{J}_{ab}(t)\sigma_i^a\sigma_i^b\,,
\end{equation}
where the matrix $J$ is replaced by a time-dependent $\tilde{J}(t)$, defined by 
\begin{equation}\label{e:JPic}
  \tilde{J}(t)= B(th/\hbar)^\trsp J B(th/\hbar)\,.
\end{equation}
Define $b=th/\hbar$, $\hat{h}=h/\abs{h}$, and $c^0=\cos{\abs{2b}}$. Moreover, define $c^u=\hat{h}^u\sin{\abs{2b}}$, and $c^{uv}=2\hat{h}^u\hat{h}^v\sin^2{\abs{b}}$, where $u,v\in\{x,y,z\}$. Using these definitions, $B(b)$ is written as 
\begin{equation}\label{e:B}
  B(b)=\left(\!\!\!\begin{array}{ccc}
  c^{xx}+c^0 & c^{xy}+c^z & c^{xz}-c^y \\
  c^{xy}-c^z & c^{yy}+c^0 & c^{yz}+c^x \\
  c^{xz}+c^y & c^{yz}-c^x & c^{zz}+c^0 \end{array}\!\!\!\right)\,.
\end{equation}

\subsection{Coefficient expansion and scaled variables}
Convinced that a ``good'' coefficient expansion should automatically satisfy the trace properties for reduced density operators \eref{e:traceprop}, and noting that the following definitions
\numparts\begin{eqnarray}
  F_1 &=& \frac12\big( \one_1 + \sum_{a\in \setI} f_1^a\sigma_1^a\big) \,,\\
  F_{12} &=& \frac14\big( \one_{12} + \sum_{a\in\setI} f_1^a( \sigma_1^a+\sigma_2^a) + \sum_{a,b\in\setI}f_2\sigma_1^a\sigma_2^b\big)\,.
\end{eqnarray}\endnumparts
satisfy $\Tr_2 F_{12}=F_1$ and $\Tr_1F_1=1$, in \cite{paskauskas2012a} we attempted a generalization by
\begin{equation}\label{e:Fexp}
  F_{1\ldots n} = 2^{-n}\sum_{s=0}^{n}
  \sum_{a\in \mathcal{I}^s} f_s^{a}
  \sum_{p\in \setP_s(n)}\prs_{p}^{a}\,,
\end{equation}
where $\sigma_i^a$ is the $a$-th component of the Pauli operator on the $i$-th lattice site. Here $\{f_s\}_{s=0}^n$ is the set of coefficients of expansion of $F_{1\dots n}$, where $f_s\in\RR^{3\times\dots\times3}$ is a rank-$s$ symmetric tensor with components $f_s^{a}$, and where $a=(a_1a_2\dots a_s)\in \setI^s$. By convention, $f_0=1$. Expansion \eref{e:Fexp} holds also for $\tilde{F}_{1\dots n}$ in terms of coefficients 
$\{ \tilde{f}_s \}_{s=0}^n$. 

The set $\setP_s(n)$ consists of all $s$-element permutations $\left(p_1,\ldots,p_s\right)$ of particle labels $p_i\in\{1,\dots,n\}$ such that $p_i<p_j$ for all $1\le i<j\le n$ (i.e.\ all sequences are strictly increasing.) Considering for example $s=2$, we have $\setP_2(n)=\{(1,2),(1,3),\ldots,(1,n),(2,3),\ldots,(n-1,n)\}$. Lastly, a product of $s$ Pauli operators on different lattice sites is denoted by 
\begin{equation}\label{e:pp}
  \prs_p^a=\prod_{i=1}^s \sigma_{p_i}^{a_i}\,.
\end{equation}

The expansion coefficients are related to the expectation values of many-particle spin observables by the identity
\begin{equation}
  f_s^a = \avg{\prs^a},
\end{equation}
where $\avg{\prs^a}$ is the expectation value of a product of $s$ spin operators, defined by
\begin{equation}
  \avg{\prs^a}=\Tr_{1\dots s} \big( F_{1\dots s}  \prod_{i=1}^s \sigma_i^{a_i} \big)\,.
\end{equation}

Mapping of the operator {BBGKY} hierarchy \eref{e:spinkin} and \eref{e:spinkin2} onto equations for the coefficients $\{f_s\}_{s=1}^N$ and $\{\tilde{f}_s\}_{s=1}^N$ is discussed in \cite{paskauskas2012a}. Here, only the final result for the interaction picture \eref{e:spinkin2} is recalled:
\begin{equation}\label{e:bbgky0}
 \frac{\hbar}{2}\partial_t \tilde{f}_n = \lambda \tilde{v}^-_{ n}[\tilde{f}_{n-1}] + (1-n\lambda)\tilde{v}^+_{n}[\tilde{f}_{n+1}]\,,
\end{equation}
where the components of $\tilde{v}_n^\pm$ are given by 
\numparts\begin{eqnarray}
  (\tilde{v}^-_{n})^{(a_1\dots a_n)}[\tilde{f}_{n-1}]= -\sum_{b,c\in\mathcal{I}} \sum_{\scriptstyle i,j=1\atop \scriptstyle i\neq j}^ n \eps_{a_i bc} \tilde{J}_{ba_j}(t) \tilde{f}_{ n-1}^{a-a_i+c-a_j},\label{e:v-}\\
  (\tilde{v}^+_{n})^{(a_1\dots a_n)}[\tilde{f}_{n+1}]=-\sum_{b,c,d\in\mathcal{I}} \tilde{J}_{bd}(t) \sum_{i=1}^n \eps_{a_ibc}\tilde{f}_{ n+1}^{a-a_i+c+d}\,.\label{e:v+}
\end{eqnarray}
\endnumparts
Here, $\eps_{abc}$ is the Levi-Civita symbol defined according to the convention $\eps_{xyz}=1$, and
\begin{equation}
a-a_i+c-a_j=(a_1,\ldots,a_{i-1},c,a_{i+1},\ldots,a_{j-1},a_{j+1},\ldots,a_n)
\end{equation}
in \eref{e:v-} is derived from $a=(a_1\dots a_n)$ by replacing the $i$th element by $c$ and then deleting the $j$th entry of $a$, and
\begin{equation}
a-a_i+c+d=(a_1,\ldots,a_{i-1},c,a_{i+1},\ldots,a_n,d)
\end{equation}
\endnumparts
in \eref{e:v+} is obtained from $a$ by replacing the $i$th element by $c$ and then appending $d$.

Tensors $\tilde{v}_n^\pm[\tilde{f}_{n\pm1}]$ are linear in the tensors $\tilde{f}_{n\pm1}$ and in the time-dependent matrix $\tilde{J}(t)$. Moreover, $\tilde{v}_n^\pm$ do not depend on $\lambda$ explicitly. A change of  $t$ and $\tilde{f}_n$ variables to
\begin{equation}
  \tau=2t\lambda^{1/2}/\hbar, \qquad \tilde{f}'_n=\tilde{f}_n/\lambda^{n/2}\,\label{e:scaling}
\end{equation}
balances all the powers of $\lambda$ in the coefficients of $\tilde{v}_n^\pm$ in \eref{e:bbgky0}. These coefficients are then cancelled from the equation with the goal to arrive at a scale-free {BBGKY} hierarchy. The remaining dependence on the number of particles enters the scaled equations through the coefficient $1-n\lambda$ at order $n$, and in the matrix $\tilde{J}(t)$. By \eref{e:JPic}--\eref{e:B}, $\tilde{J}(t)$ is a sum of a constant term $\oJ$ and time-dependent terms which, by the form of $B(th/\hbar)$, are trigonometric functions such as $\cos{(2t\abs{h}/\hbar)}$. In $\tau$ units of time, such terms are a fast oscillation on top of the mean $\oJ$, and it was argued \cite{paskauskas2012a} that these fast oscillations can be taken into account by averaging (which becomes exact in thermodynamic limit.) With the additional assumption that $\tilde{f'}$ varies slowly in the period of oscillation $T_\lambda$, $\tilde{f'}$ is replaced by its average value $\overline{f}'(\tau)=(T_\lambda)^{-1}\int_\tau^{\tau+T}\tilde{f'}(\tau')d\tau'$. Application of this procedure on \eref{e:bbgky0} results in the scaled and averaged {BBGKY} hierarchy
\begin{equation}\label{e:bbgky1}
  \partial_\tau \overline{f}'_n\approx\overline{v}^-_{n}[\overline{f}'_{n-1}] + (1-\lambda n)\overline{v}^+_{n}[\overline{f}'_{n+1}]\,,
\end{equation}
where $\overline{v}_n^\pm$ are obtained from $\tilde{v}_n^\pm$ by replacing $\tilde{J}(t)$ by $\oJ$. The latter is defined by
\begin{equation}\label{e:Javg}
  \oJ =T^{-1} \int_0^{T} \tilde{J}(t) dt\qquad T=\pi \hbar/\abs{h}\,.
\end{equation}
 The approximation in \eref{e:bbgky1} is a result of a substitution $\tilde{f}\to \overline{f}'$. It turns into an equality in two cases. Firstly, if there is no magnetic field in \eref{e:ham} i.e. if $h=0$, then there is no need to perform averaging. In this case the non-averaged coefficient {BBGKY} hierarchy has the same form as \eref{e:bbgky1} where $\oJ$ is equal to the original coupling matrix $J$. Secondly, if separation of scales is a valid assumption\footnote{It is justified rigorously in a special case, treated in \cite{paskauskas2012a}}, the averaging is exact in the limit $\lambda=0$. 

This limit describes thermodynamic evolution of the model by means of an infinite dimensional set of differential equations. The structure of these equations is the main focus of this article. Dropping the bars and the primes from $\overline{f}'$ and $\overline{v}_n^\pm$ for simplicity, the thermodynamic limit {BBGKY} hierarchy is defined as
\begin{equation}\label{e:bbgky2}
  \partial_\tau f_n=v^-_{n}[f_{n-1}] + v^+_{n}[f_{n+1}]\,.
\end{equation}
The distinction between the original matrix $J$ and the averaged one, $\oJ$, is maintained.

\section{Desymmetrized coefficient {BBGKY} hierarchy}\label{s:equations}
This section introduces the desymmetrized coefficient {BBGKY} hierarchy and presents the framework for the subsequent study of its properties in sections \ref{s:transformation}--\ref{s:symmetries}.
\subsection{Coefficient equations}\label{s:desym}
\Eref{e:Fexp} maps $F_{1\dots N}$ onto a set of tensor coefficients $\{f_n\}_{n=0}^N$, where $f_n$ is a rank-$n$ three dimensional symmetric tensor with components $f_n^a$, and where $a=(a_1\dots a_n)$ is a coordinate multiindex, and $a_i\in\{x,y,z\}$. A symmetric rank-$n$ tensor in three dimensions is specified completely by $d(n)$ numbers, where
\begin{equation}\label{e:dn}
  d(n) = \sum_{j=1}^{n+1}j = (n+1) (n+2) / 2.
\end{equation}
The number of degrees of freedom $D(N)$ is defined as the number of independent coefficients in the expansion of $F_{1\dots N}$. It is equal to a sum of the numbers of independent components from all the tensors $f_n$ excluding the trivial $f_0=1$,
\begin{equation}\label{e:Dn}
  D(N)=\sum_{n=1}^N d(n) = (N+1)(N+2)(N+3)/6-1.
\end{equation}
A sequence, formed from all the independent components of $f_n$, can be viewed as a vector $\mathrm{indep}(f_n)$ on a $d(n)$-dimensional vector space $\hilb_n$,
\begin{equation}
  \mathrm{indep}(f_n) \in \hilb_n\,.
\end{equation}
where $\hilb_n\subseteq\RR^{d(n)}$ is called the $n$-spin desymmetrized coefficient vector space, or simply the $n$-spin vector space. A sequence of all independent tensor components
\begin{equation}\label{e:f}
  \mathrm{indep}(f) = ( \mathrm{indep}(f_1), \dots, \mathrm{indep}(f_N) ) \in \Hilb_N
\end{equation}
can be viewed as a vector on a direct sum of $n$-spin desymmetrized vector spaces,
\begin{equation}\label{e:hilbD}
  \Hilb_N=\bigoplus_{n=1}^N\hilb_n\,.
\end{equation}

By the symmetry of $f_n$, any $f_n^{a'}$ is equal to $f_n^{a(m)}$, where
\begin{equation}\label{e:a(m)}
  a(m)=(\underbrace{x,\ldots,x}_{\hspace{-3mm}\displaystyle m^x\mbox{ elements}\hspace{-3mm}},y,\ldots,y,z,\ldots,z)
\end{equation}
is a permutation of $a'$ with the components $x$, $y$, $z$ arranged in contiguous blocks of $m^x$, $m^y$, and $m^z$ labels. Hence, all independent components of $f_n$ can be enumerated uniquely by a triple of nonnegative integers $m\in \nodes_n$, where the set $\nodes_n$ is defined by 
\begin{equation}\label{e:setMn}
  \nodes_n = \{ (m^x,m^y,m^z): m^a\ge0, m^x+m^y+m^z=n\}\,.
\end{equation}
In order to denote the allowed set of indices for an arbitrary $n$, define
\begin{equation}
  \nodes = \{ (m^x,m^y,m^z): m^a\ge0, m^x+m^y+m^z\le N\}\,.
\end{equation}
Define the norm of $m\in\nodes_n$ by
\begin{equation}\label{e:mnorm}
  \abs{m} = m^x+m^y+m^z\,.
\end{equation}
As a result, all $\mathrm{indep}(f)\in\Hilb_N$ are uniquely labeled by elements of
\begin{equation}\label{e:setM}
  \nodes^{(N)}=\bigcup_{n=1}^N \nodes_n\,.
\end{equation}
In the subsequent, instead of a cumbersome $\mathrm{indep}(f_n)$, the sequence of all the independent components $f_n^{a(m)}$ will be denoted by $f_n(m)\in\hilb_n$ where $m\in\nodes_n$.

The coefficient hierarchy \eref{e:bbgky2} is translated to a desymmetrized coefficient {BBGKY} hierarchy on $\Hilb_N$. This is formally achieved by replacing $f_n^{a'}$ by $f_n(m)\equiv f_n^{a(m)}$, and by replacing $v_n^\pm[f_{n\pm1}]$ on the right-hand side of \eref{e:bbgky2} by $v_n^\pm(m, f_{n\pm1})$
\begin{equation}\label{e:bbgky3}
\partial_\tau f_n(m,\tau)= v_{n}^-(m,f_{n-1}) + (1-n\lambda) v_n^+(m,f_{n+1}).
\end{equation}
 The derivation of $v_n^\pm(m,f_{n\pm1})$ is given in \ref{s:derivation}. It is convenient to separate the diagonal and off-diagonal components of $\oJ$. Define vectors $\oK=(\oK_x,\oK_y,\oK_z)$ and $\oW=(\oW_x,\oW_y,\oW_z)$ as
\numparts\begin{eqnarray}
(\oK_x,\oK_y,\oK_z) &=& (\oJ_{yy}-\oJ_{zz}, \oJ_{zz}-\oJ_{xx}, \oJ_{xx}-\oJ_{yy})\,, \label{e:K} \\
(\oW_x,\oW_y,\oW_z) &=& (\oJ_{yz}, \oJ_{zx}, \oJ_{xy}) \label{e:W} \,.
\end{eqnarray}\endnumparts
Replacing $f_{n\pm1}$ by a generic $f$ in $v_n^{\pm}(m,f_{n\pm1})$ and using the notation \eref{e:K}--\eref{e:W}, terms $v_n^\pm(m,f)$ for general $\oK$ and $\oW$, and hence a general symmetric $\oJ$, are given by
\numparts
\begin{eqnarray}
  \fl v_{n}^+(m,f) &=& - \oK_xm^x f(m^x-1,m^y+1,m^z+1) -\oK_ym^y f(m^x+1,m^y-1,m^z+1)\nonumber\\
  \fl && - \oK_zm^z f(m^x+1,m^y+1,m^z-1)\nonumber\\
  \fl  &-&\oW_x(m^y-m^z) f(m^x+1,m^y,m^z) - \oW_y(m^z-m^x)f(m^x,m^y+1,m^z) \nonumber\\
  \fl  && -\oW_z(m^x-m^y)f(m^x,m^y,m^z+1) \nonumber \\
  \fl  &-& m^x\oW_x\big[ f(m^x-1,m^y,m^z+2)-f(m^x-1,m^y+2,m^z) \big] \nonumber \\
  \fl  &-& m^y\oW_y\big[ f(m^x+2,m^y-1,m^z)-f(m^x,m^y-1,m^z+2) \big] \nonumber \\
  \fl  &-& m^z\oW_z\big[ f(m^x,m^y+2,m^z-1)-f(m^x+2,m^y,m^z-1) \big]\,. \label{e:vplus}\\
  \fl \tilde{v}_{n}^-(m;f) &=& \oK_xm^ym^z f(m^x+1,m^y-1,m^z-1) + \oK_ym^zm^x f(m^x-1,m^y+1,m^z-1) \nonumber \\
  \fl  && +  \oK_zm^xm^y f(m^x-1,m^y-1,m^z+1)                                     \nonumber \\
  \fl  &+& \oW_xm^x(m^y-m^z) f(m^x-1,m^y,m^z) +  \oW_ym^y(m^z-m^x) f(m^x,m^y-1,m^z)  \nonumber \\
  \fl  && + \oW_zm^z(m^x-m^y) f(m^x,m^y,m^z-1)\nonumber \\
  \fl  &+& m^x(m^x-1) \big[ \oW_y f(m^x-2,m^y+1,m^z)-\oW_zf(m^x-2,m^y,m^z+1)\big] \nonumber \\
  \fl  &+& m^y(m^y-1) \big[ \oW_z f(m^x,m^y-2,m^z+1)-\oW_xf(m^x+1,m^y-2,m^z)\big] \nonumber \\
  \fl  &+& m^z(m^z-1) \big[ \oW_x f(m^x+1,m^y,m^z-2)-\oW_yf(m^x,m^y+1,m^z-2)\big]\,.\label{e:vminus}
\end{eqnarray}\endnumparts
In \eref{e:vplus}--\eref{e:vminus}, $\mathrm{indep}(f)$ has been replaced by $f$, and in the subsequent the notation $f\equiv f(m)$ will imply an element of $\Hilb_N$.

\subsection{Permutation invariant form of equations}\label{s:permutation}
The {BBGKY} hierarchy \eref{e:bbgky1} can be written in an explicitly permutation invariant form, which is more concise and facilitates further analysis. 

Define a cyclic permutation operator, acting on $\setI= \{x,y,z\}$, by 
\begin{equation}
  \pi (x,y,z) = (y,z,x)
\end{equation}
and three reflection operators $r_{x}$, $r_{y}$, $r_{z}$ which act by permuting pairs of elements. For example
\begin{equation}
  r_{x}(x,y,z) = (x,z,y)
\end{equation}
The group of permutations in three dimensions consists of six elements, ${\mathfrak g}=\{i, \pi, \pi^2, r_{1}, r_{2}, r_{3}\}$, where $i$ is the identity operator $i(x,y,z)=(x,y,z)$.

Define twelve vectors $t_i^a$, labelled by $i=1,2,3,4$ and $a\in\{x,y,z\}$, by
\begin{equation}\label{e:tai}
  \displaystyle
  \begin{array}{lll}
    t_1^x=(-1,1,1),\quad & t_1^y=(1,-1,1), \quad & t_1^z=(1,1,-1). \vspace{0.1cm}\\
    t_2^x=(1,0,0), &  t_2^y=(0,1,0), & t_2^z=(0,0,1). \vspace{0.1cm}\\
    t_3^x=(-1,0,2), & t_3^y=(2,-1,0), & t_3^z=(0,2,-1).\vspace{0.1cm}\\
    t_4^x=(-1,2,0), & t_4^y=(2,0,-1), & t_4^z=(0,-1,2).
  \end{array}
\end{equation}
Note that $t_i^y=\pi t_i^x$, $t_i^z=\pi t_i^y=\pi^2t_i^x$, and $t_4^a=r_at_3^a$. Using these definitions, the coefficient {BBGKY} hierarchy equations \eref{e:bbgky3} and \eref{e:vplus}--\eref{e:vminus} can be written as
\begin{eqnarray}\label{e:graph}
\fl \partial_\tau f_n(m) &=& \sum_{a\in \setI}\oK_a\Bigl\{ m^{\pi a}m^{\pi^{2}a} f_{n-1}(m-t_1^a) - m^a f_{n+1}(m+t_1^a)\Bigr\}\nonumber\\
\fl  &+& \sum_{a\in \setI}\oW_a(m^{\pi a}-m^{\pi^2a}) \Bigl\{ m^af_{n-1}(m-t_2^a) - f_{n+1}(m+t_2^a)\Bigr\}\nonumber\\
\fl  &+& \sum_{a\in \setI}\oW_a\Bigl\{m^{\pi^2a}(m^{\pi^2a}-1) f_{n-1}(m-t_3^a) - m^{\pi a}(m^{\pi a}-1)f_{n-1}(m-t_4^a) \nonumber \\
\fl  && - m^a\bigl[ f_{n+1}(m+t_3^a)-f_{n+1}(m+t_4^a)\bigr]\Bigr\}
\end{eqnarray}

\subsection{Representation by a graph}\label{s:graphs}
To study the properties of \eref{e:graph} it is convenient to use the language of the graph theory. A graph $\graph=(\nodes,\edges)$ is an abstract representation of two sets: The set $\nodes$ of objects called nodes $m\in\nodes$ will be identified with the set of indices \eref{e:setM} of elements of $\Hilb_N$. Certain pairs of nodes $m$ and $m_1$ of the graph will be connected by links, called edges and denoted by $\vv{m,m_1}$. The set of all the edges is denoted by $\edges$. 

Couplings among tensor coefficients in \eref{e:graph} [or in \eref{e:bbgky2} or \eref{e:bbgky3} which are equivalent to \eref{e:graph}] will be allowed to determine $\edges$ by the following definition. 
\begin{definition}\label{def:edge} An edge $\vv{m,m_1}$ is an ordered pair of elements $m, m_1\in\nodes$. $\vv{m,m_1}\in\edges$ if and only if two conditions are satisfied: $\partial_\tau f_n(m)$ is coupled with $f_{n_1}(m_1)$ by \eref{e:graph} and the coupling coefficient is not equal to zero. 
\end{definition}
Two more definitions will help establish the terminology. 
\begin{definition}\label{def:adjoint}
  The node $m_1$ is said to be adjacent to $m$ if $\vv{m,m_1}\in\edges$. A set $\adj (m)\subset\nodes$ is a set of all of the nodes, adjacent to $m$.
\end{definition}
\begin{definition}\label{def:value}
  A value function $\psi: \edges\to \RR$ is equal to the coupling coefficient between $\partial_\tau f_n(m)$ and $f_{n_1}(m_1)$ and is denoted by  $\psi[\vv{m,m_1}]$. 
\end{definition}
The three-term recurrence structure of the {BBGKY} hierarchy implies that $\adj (\nodes_n)\subseteq\{ \nodes_{n-1},\nodes_{n+1}\}$. The adjacency relations will be labelled by $\pm$ to distinguish this structure. 
\begin{definition}\label{def:adj}
  A set $\adj^\pm(m)\subset\nodes$ is a set of all the nodes $m_1$ such that $m_1$ is adjacent to $m$ and 
$\abs{m_1} = \abs{m}\pm 1$. 
\end{definition}
By this definition,
\begin{equation}\label{e:adj}
  \adj^\pm(\nodes_n) \subseteq \nodes_{n\pm1}. 
\end{equation}

Having established the terminology of graph theory, let us now determine the implications of the structure of the {BBGKY} hierarchy \eref{e:graph} on the structure of the graph $\graph$. For this purpose the following two lemmas are provided. 

Assume that $m=(m^x, m^y, m^z)\in\nodes$, $m_1=(m_1^x,m_1^y, m_1^z)\in\nodes$, $m_1\in\adj^+(m)$. 
\begin{lemma}\label{l:upedge}
  Nodes $m$ and $m_1$ are connected by a edge $\vv{m,m_1}\in\edges$ iff for some $a\in\setI$ both conditions in one of \eref{e:t1}--\eref{e:t4} are satisfied
   \numparts\begin{eqnarray}
   &m_1=m+t_1^a  &\psi[\vv{m,m_1}]=-\oK_am^a\neq0.\label{e:t1}\\
   &m_1=m+t_2^a  &\psi[\vv{m,m_1}]=-\oW_a(m^{\pi a}-m^{\pi^2a})\neq0.\label{e:t2}\\
   &m_1=m+t_3^a  &\psi[\vv{m,m_1}]=-\oW_am^a\neq0.\label{e:t3}\\
   &m_1=m+t_4^a\qquad  &\psi[\vv{m,m_1}]=\oW_am^a\neq0.\label{e:t4}
   \end{eqnarray}\endnumparts
\end{lemma}
\begin{proof}
  By \defref{def:edge} and by assumption $m_1\in\adj^+(m)$, an edge exists iff the term $\partial_\tau f_n(m)$ in \eref{e:graph} has a non-zero coupling with $f_{n+1}(m_1)$.  Since $\forall a,i:~\abs{m\pm t_i^a}=\abs{m}\pm1$, the first necessary condition for $\vv{m,m_1}\in\edges$ is that $m_1$ is equal to one of the terms listed in the first column of \eref{e:t1}--\eref{e:t4}. The second necessary condition is non-vanishing of the corresponding coupling coefficient, which are listed in the second column of \eref{e:t1}--\eref{e:t4}.
\end{proof}
\begin{lemma}\label{l:downedge}
  $\vv{m,m_1}\in\edges$ implies $\vv{m_1,m}\in\edges$.
\end{lemma}
\begin{proof}
  Assume that $\vv{m,m_1}\in\edges$. By \lemref{l:upedge}, $m_1=m+t_i^a$ for some $i=1,2,3,4$ and some fixed $a\in\setI$. By \defref{def:adj}, the adjacent set $\adj^-{(m_1)}$ is determined by the non-zero couplings between $\partial_\tau f_{n+1}(m_1)$ and $f_n(m_2)$ where $m$ is substituted by $m_1$ in \eref{e:graph}. By \eref{e:graph} and \lemref{l:upedge}, $m_2=m+t_i^a-t_j^b$, therefore $m_2=m$ if $i=j$ and $a=b$. Examining the coupling coefficients for each $i$ and $a$ in \eref{e:graph} one finds
\numparts\begin{eqnarray} 
&\psi[\vv{m_1,m}]=\oK_a(m^{\pi a}+1)(m^{\pi^2a}+1)\qquad &i=1. \label{e:psi1} \\
&\psi[\vv{m_1,m}]=\oW_a(m^a+1)(m^{\pi a}-m^{\pi^2 a})    &i=2. \label{e:psi2} \\
&\psi[\vv{m_1,m}]=\oW_a(m^{\pi^2a}+1)(m^{\pi^2a}+2)      &i=3. \label{e:psi3} \\
&\psi[\vv{m_1,m}]=-\oW_a(m^{\pi a}+1)(m^{\pi a}+2)      &i=4. \label{e:psi4}
\end{eqnarray}\endnumparts
Denote the absolute value of an integer coefficient of $\psi[\vv{m,m_1}]$, multiplying a coupling constant $\oK_a$ or $\oW_a$, by $\nu(m,m_1)$, and the one of $\psi[\vv{m_1,m}]$ by $\nu(m_1,m)$. By assumption that $\psi[\vv{m,m_1}]\neq0$, it follows that the relevant coupling constant $\oK_a\neq0$ or $\oW_a\neq0$, and $\nu(m,m_1)\neq0$. Comparing \eref{e:t1}--\eref{e:t4} with \eref{e:psi1}--\eref{e:psi4} and using $m^a+1>0$, it follows that under these conditions $\nu(m_1,m)\neq0$ as well.
\end{proof}

Comparing \eref{e:t1}--\eref{e:t4} with \eref{e:psi1}--\eref{e:psi4} we conclude that $\nu(m,m_1)\neq \nu(m_1,m)$ in general. In particular,
\begin{equation}\label{e:nu}
  -\frac{\psi[\vv{m,m_1}]}{\psi[\vv{m_1,m}]} = \frac{\nu(m,m_1)}{\nu(m_1,m)}
\end{equation}
is equal to a ratio of two positive integers.
For $m_1=m+t_i^a$, define $\gamma_i^a$ by
\begin{equation}\label{e:gai}
  \gamma_i^a =  -\psi[\vv{m,m+t_i^a}] \psi[\vv{m+t_i^a,m}].
\end{equation}
The following formulae for $\gamma_i^a$ are derived from lemmas \ref{l:upedge}--\ref{l:downedge}
\numparts\begin{eqnarray}
  \gamma_1^a = \oK_a^2\max{(m^x,m_1^x)}\max{(m^y,m_1^y)}\max{(m^z,m_1^z)}\,,\label{e:gamma1}\\
  \gamma_2^a = \oW_a^2(m^a+1)\left(m^{\pi a}-m^{\pi^2a}\right)^2\,,\label{e:gamma2}\\
  \gamma_{3}^a = \oW_a^2m^a(m^{\pi^2 a}+1)(m^{\pi^2a}+2)\,,\label{e:gamma3}\\
  \gamma_{4}^a = \oW_a^2m^a(m^{\pi a}+1)(m^{\pi a}+2)\,.\label{e:gamma4}
\end{eqnarray}\endnumparts

The graph $\graph$ is drawn as points and arrows on a plane. Since edges are directional, an edge $\vv{m,m_1}$, connecting the $m$-th and $m_1$-st nodes can be represented by an arrow, starting at the point $m$ and ending at $m_1$. However, by \lemref{l:downedge}, all arrows come in pairs. Therefore, instead  of connecting each adjacent pair of nodes by a pair of arrows, each pair of arrows can be substituted by a single line and an integer number on each side of the line to distinguish the direction of a link. The following convention will be adopted. For a pair of adjacent nodes $m$ and $m_1\in\adj^+(m)$, $\nu(m,m_1)$ is to the right of (or below), and $\nu(m_1,m)$ is to the left of (or above) the line, connecting $m$ and $m_1$.

To take into account the adjacency relation \eref{e:adj} and the composite structure \eref{e:setM}, all the nodes $m\in\nodes_n$ with a fixed $n$ will be drawn at the same height. For example, a segment of $\graph^1$ \footnote{Some nodes and vertices of $\graph^1$ are not drawn in order to keep the image clear. } is shown in \eref{e:Gx}, where each node is denoted by $m^x.m^ym^z$. 
\begin{equation}\label{e:Gx}
  \newcommand\av[2]{\ar@{-}@<0.15ex>[u]^{#1}_{#2}}
  \newcommand\ad[2]{\ar@{-}@<0.15ex>[ur]^{#1}_{#2}}
  \newcommand\add[2]{\ar@{-}@<0.15ex>[urr]^{#1}_{#2}}
  \newcommand\lk{\ar@{-}@<0.15ex>[ul];}
  \xymatrix@C=1pc@R=1pc{
    1.06 &1.24 &3.04&3.22&5.02 &7.00\\
    0.15\av{6}{1}\ad{2}{5}&2.13\av{8}{2}\ad{}{}\add{6}{3}&0.33\lk&4.11\av{4}{4}\ad{10}{1}&\\
    1.04\av{5}{1}\ad{2}{4}&1.22\av{6}{2}\ad{9}{1}&3.02\ad{4}{2}&5.00\av{1}{5}\\
    0.13\av{4}{1}\ad{2}{3}&2.11\av{4}{2}\ad{6}{1}&\\
    1.02\av{3}{1}\ad{2}{2}&3.00\av{1}{3}\\
    0.11\av{2}{1} &\\
    1.00\av{1}{1}
  }
\end{equation}

\section{Initial value problem for the coefficient {BBGKY} hierarchy}\label{s:transformation}

\subsection{Generalized Liouville operator}
Time evolution of coefficients $f$ is obtained by solving the initial value problem
\begin{equation}\label{e:iv}
  \partial_\tau f(\tau,f_{\mathrm{i}}) = \liouv f(\tau,f_{\mathrm{i}})\,,\qquad f(0,f_{\mathrm{i}})=f_{\mathrm{i}},
\end{equation}
where $\liouv f(\tau, f_{\mathrm{i}})$ denotes the right-hand side of the {BBGKY} hierarchy \eref{e:bbgky3}. The initial condition $f_{\mathrm{i}}$ is determined from the initial condition for the operator $F_{1\dots N}$. 

The linear operator $\liouv: \Hilb_N\mapsto\Hilb_N$ is called the generalized Liouville operator \cite{balescu1975}. It is represented by a real-valued matrix, which multiplies vectors $f=(f_1,\dots,f_N)^\trsp$. The latter consist of $N$ blocks by the composition formula \eref{e:f}. The $n$-th block of $f$ consists of all the coefficients in the $n$-spin vector space arranged in a sequence. This block structure of the vector space $\Hilb_N$ imparts a similar block structure on $\liouv$, comprising $N\times N$ rectangular matrices $L^{(ij)}\in\RR^{d(i)\times d(j)}$, where $1\le i,j\le N$. The three-term recurrence structure of \eref{e:bbgky3} implies that the only non-zero blocks $L^{(ij)}$ are such blocks that $\abs{i-j}=1$. Define 
\begin{equation}\label{e:Ln}
\fl  L_n^+ = -L^{(n,n+1)}\in\RR^{d(n)\times d(n+1)}\,,\qquad L_n^-= L^{(n+1,n)}\in\RR^{d(n+1)\times d(n)}.
\end{equation}
The generalized Liouville operator $\liouv$ has the form of a block tridiagonal matrix with rectangular off-diagonal blocks $L_n^+$ and $L_n^-$
\begin{equation}\label{e:matrix}
  \liouv = \left(\!\!\!\begin{array}{cccccc}
  0 & -L^+_1 & & & & \\
  L_1^- & 0 & -L_2^+ & & & \\
  & L_2^- & 0 & & &\\
  &   &  & \ddots & & \\
  &   &  &        &0& -L_{N-1}^+ \\
  &   &  &        & L_{N-1}^- & 0
  \end{array}\!\!\!\right).
\end{equation}

\subsection{Self-adjoining transformation}
The {BBGKY} hierarchy \eref{e:bbgky3} can be rewritten using the graph notation as 
\begin{equation}\label{e:graph2}
  \fl \partial_\tau f_n(m) = \sum_{m_1\in\adj^-(m)}\psi[\vv{m,m_1}] f_{n-1}(m_1) +  \sum_{m_1\in\adj^+(m)}\psi[\vv{m,m_1}]f_{n+1}(m_1)\,.
\end{equation}
Comparing this equation with \eref{e:matrix} we find that for any $m\in\nodes_n$ and $m_1\in\nodes_{n+1}$, one matrix element of $L_n^+$ is equal to $-\psi[\vv{m,m_1}]$ and similarly, one element of $L_n^-$ is equal to $\psi[\vv{m_1,m}]$. Relations between the indices $i$ and $j$ of a particular element $(L_n^\pm)_{ij}$ and value functions $\psi$ are determined by an arbitrary choice of a mapping of $m\in\nodes_n$ onto a sequence of integers $\{1,\dots, d(n)\}$ for each $n$. It is not difficult to see that $(L_n^+)_{ij}=-\psi[\vv{m,m_1}]$ for some $i$ and $j$ implies $(L_n^-)_{ji}=\psi[\vv{m_1,m}]$ for the same pair $i$, $j$. By a corollary of \lemref{l:downedge} it follows that in general $(L_n^+)_{ij}\neq(L_n^-)_{ji}$, i.e. that $L_n^\pm$ are not mutually adjoint, and as a result that $\liouv$ is not anti-Hermitian
\begin{equation}\label{e:notaherm}
  L_n^-\neq L_n^{+\trsp} \quad \rightarrow \quad \liouv \neq - \liouv^\dagger\,.
\end{equation}
The evolution operator is formally obtained by exponentiating $\liouv$,
\begin{equation}\label{e:green}
  \green(\tau) = \rme^{\liouv \tau}.
\end{equation}
As a consequence of \eref{e:notaherm}, $\green(\tau)$ does not preserve the conventional inner product $(f,g)=\sum f_ig_i$ of $f, g\in\Hilb_N$ for infinitesimal $\tau$ and therefore $\green(\tau)$ is not unitary with respect to this inner product. To get a positive result, one could try a change of basis and a similarity transformation of $\liouv$ by an invertible ${\sf D}$ such that $\aherm$, defined by 
\begin{equation}\label{e:aherm}
  {\aherm} = {\sf D}^{-1} \liouv {\sf D}
\end{equation}
is anti-Hermitian,
\begin{equation}\label{e:aherm2}
  \aherm^\dagger=-\aherm\,.
\end{equation}
Since for a tridiagonal matrix (a matrix of the form \eref{e:matrix} where $L_n^\pm$ are scalars) this can be achieved by a scaling transformation, consider $\sf D$, defined by 
\begin{equation}\label{e:diagD}
  {\sf D} = {\mathrm{diag}}( D_1, D_2, \dots, D_{N} )\,,
\end{equation}
where $D_i$ is a $d(i)\times d(i)$ square matrix which is yet to be found. The transformation \eref{e:aherm} by \eref{e:diagD} preserves the block tridiagonal structure of ${\liouv}$. In analogy with \eref{e:Ln}--\eref{e:matrix}, denote the off-diagonal blocks of ${\aherm}$ by $H_n^\pm$. By substituting into \eref{e:aherm}, the condition \eref{e:aherm2} is equivalent to
\begin{equation}\label{e:Hn}
  H_n^{+}= D_n^{-1} L_n^+ D_{n+1}\,, \qquad   H_n^{-}= D_{n+1}^{-1} L_n^- D_{n}\,.
\end{equation}
The condition that ${\aherm}$ is anti-Hermitian is
\begin{equation}\label{e:antihermitian}
  H_n^{-}=H_n^{+\trsp}
\end{equation} 
and from \eref{e:Hn} and \eref{e:antihermitian} it follows that \eref{e:aherm2} is equivalent to existence of a solution of 
\begin{equation}\label{e:symcond}
  (D_{n+1}D_{n+1}^\trsp) L_n^{+\trsp}  = L_n^-(D_nD_n^\trsp)\,, \qquad n=1,\dots, N-1,
\end{equation}
where $D_{n}D_{n}^\trsp$ must be determined by solving \eref{e:symcond} starting from an arbitrarily chosen $D_1$. For each $n$, \eref{e:symcond} is a set of $d(n)\times d(n+1)$ equations for the unknown elements of a square symmetric matrix, which has $d(n+1)[d(n+1)+1]/2$ free parameters. A necessary condition for solvability of \eref{e:symcond} is a greater number of unknowns than the number of equations, which amounts to satisfying $[ d(n+1)+1 ]/2 \ge d(n)$. Using \eref{e:dn} this inequality simplifies to $2n+6\ge0$ which is certainly satisfied for $n>0$. To construct an explicit form of $\aherm$, the following theorem is provided. 
\begin{theorem}\label{the:transformation}
  Equations \eref{e:aherm}--\eref{e:aherm2} are solvable by a diagonal matrix ${\sf D}$. Each element of $H_n^\pm$ of $\aherm$ is equal to $\sqrt{\gamma_i^a}$, defined by \eref{e:gai}, for some $i$ and $a$.
\end{theorem}
\begin{proof}
  Assume $c\in\Hilb_N$ such that $c(m)\neq0\forall m\in\nodes$. Consider a transformation of $f$
\begin{equation}\label{e:f'}
  f'_n(m) = c(m) f_n(m)
\end{equation}
This transformation is of the form \eref{e:aherm} with a diagonal $\sf D$, where ${\sf D}_{ii}=1/c(m)$ $\forall i$ and some $m=m(i)\in\nodes$.
It is convenient to express the {BBGKY} hierarchy equations for $f'_n(m)$ in the form \eref{e:graph2}. Substitution of \eref{e:f'} into \eref{e:graph2} results in the equation of the same form as \eref{e:graph2} with $f$s replaced by $f'$s, and $\psi$s replaced by $\psi'$s, and where the latter are defined by
\begin{equation}\label{e:psi'}
  \psi'[\vv{m,m_1}] = \psi[\vv{m,m_1}] c(m)/c(m_1)
\end{equation}
The proof will be established if $c(m)$ can be determined from \eref{e:psi'}, and the condition 
\begin{equation}\label{e:psi'cond}
  \psi'[\vv{m,m_1}]=-\psi'[\vv{m_1,m}]
\end{equation}
is satisfied for each $m\in\nodes$ and $m_1\in\adj(m)$, as this condition is equivalent to \eref{e:aherm2}.

Combining \eref{e:psi'} and \eref{e:psi'cond}, and using \eref{e:nu} one obtains
\begin{equation}\label{e:c2}
  c^2(m) = c^2(m_1)\frac{\nu(m_1,m)}{\nu(m,m_1)}.
\end{equation}
From \eref{e:c2} it follows that $c(m)$ is determined by its adjacent node $m_1$. \Eref{e:c2} can be applied recursively to determine $c(m)$ by $c(m_1)$, where $m_1$ is any node such that $m$ is reachable from $m_1$, multiplied by a ratio of products of numerical factors $\nu$ along the path, connecting $m_1$ with $m$. 

\begin{definition}\label{d:path}A path $p$ of length $k>0$ is an ordered sequence of $k$ edges $\vv{m_j,m_{j+1}}$, $1\le j\le k$ such that $m_{j\pm1}$ is adjacent to $m_j$,
\begin{equation}
  p(m_1\rightarrow m_{k+1}) = ( \vv{m_1,m_2}, \vv{m_2,m_3}, \dots, \vv{m_{k}, m_{k+1}} ).
\end{equation}
A loop $p_0$ is a closed path $p$, i.e. a path where $m_{k+1}=m_1$
\begin{equation}
  p_0 = (\vv{m_1, m_2}, \vv{m_2,m_3}, \dots, \vv{m_{k},m_1}).
\end{equation}
\end{definition}
Using this definition
\begin{equation}\label{e:path}
  c^2(m_{k+1}) = c^2(m_1) \prod_{e\in p(m_1\rightarrow m_{k+1})} \frac{\nu(m'(e),m''(e))}{\nu(m''(e),m'(e))}\,,
\end{equation}
where $e$ denotes an element of $p$ such that $e=(m'(e),m''(e))$. Since most $m'$s have more than one adjacent $m''$, most pairs of connected nodes can be linked by several distinct paths. Therefore, solvability of \eref{e:aherm} by a diagonal $\sf D$ rests upon the condition that $c(m_{k+1})$ is uniquely determined by $m_1$ and is independent of $p$. Looking at a segment of $\graph^1$ in \eref{e:Gx} it is clear that any path $p(m_1\rightarrow m_{k+1})$ could be augmented by inserting loops at any of the intermediate nodes. Therefore, independence on the traversal is achieved if for any loop $p_0$
\begin{equation}\label{e:cond1}
  \prod_{e\in p_0} \frac{\nu(m'(e),m''(e))}{\nu(m''(e),m''(e))} =1\,.
\end{equation}
This condition could be turned into an additive condition for loops by taking a logarithm of \eref{e:cond1} and defining the current $j$ on each edge.
\begin{definition}
  For any $m'\in\nodes$ and $m''\in\adj(m')$, and $e=(m',m'')\in\edges$, the current $j(e)\in\RR$ is defined by
  \begin{equation}\label{e:j}
    j(e) =     \log\big[\nu(m',m'')/\nu(m'',m')\big]\,.
  \end{equation}
\end{definition}
Solvability of \eref{e:aherm} by a scaling transformation can be formulated as a ``zero circulation'' condition: For any loop $p_0$ we require that
\begin{equation}\label{e:jcirc}
  \sum_{e\in p_0} j(e) = 0\,.
\end{equation}
Considering lines that make up any loop in \eref{e:Gx} as boundaries of that loop that separate the ``inside'' from the ``outside'', the zero circulation conditions is verified by multiplying all the numbers along the loop which are either outside or inside its boundary. As can be seen from \eref{e:Gx}, the most elementary non-trivial loop $p_0$ involves four nodes, and more complicated loops can be constructed from a sequence of elementary loops. It is therefore sufficient to verify the circulation condition \eref{e:jcirc} for each loop of the form
\begin{equation}\label{e:loop}
  \newcommand\av[2]{\ar@{-}@<0.15ex>[u]^{#1}_{#2}}
  \newcommand\ak[2]{\ar@{-}@<0.15ex>[l]^{#1}_{#2}}
  \newcommand\ad[2]{\ar@{-}@<0.15ex>[ur]^{#1}_{#2}}
  \newcommand\add[2]{\ar@{-}@<0.15ex>[urr]^{#1}_{#2}}
  \newcommand\lk{\ar@{-}@<0.15ex>[ul];}
  \xymatrix@dr@C=2pc@R=2pc{
    m+t_i^a+t_j^b  & m+t_j^b \ak{\nu'_2}{\nu_2} \\
    m+t_i^a \av{\nu_3}{\nu'_3} & m\ak{\nu_4}{\nu'_4}\av{\nu'_1}{\nu_1} }
  \qquad \prod_i \nu_i = \prod_i \nu'_i\,.
\end{equation}
where $t_i^a\neq t_j^b$. As an example, consider in \eref{e:Gx} a loop, that visits $\{102,211,122,013\}$. We find 
$\prod_i \nu_i =2\cdot2\cdot2\cdot3=24$, and $\prod_i \nu'_i =2\cdot4\cdot3\cdot1=24$. Therefore \eref{e:loop} is satisfied for this particular loop. A straightforward calculation (not detailed here) confirms that \eref{e:loop} is satisfied for arbitrary $t_i^a\neq t_j^b$. 

The elements of $\aherm$ are expressed in terms of value functions $\psi'[\vv{m,m_1}]$. Substituting \eref{e:c2} into \eref{e:psi'}, with $k=1$, such that $m_2=m_1+t_i^a$, and using \eref{e:gai}
\begin{equation}
  \psi'[\vv{m,m+t_i^a}]=\sqrt{-\psi[\vv{m,m+t_i^a}]\psi[\vv{m+t_i^a,m}]}= \sqrt{\gamma_i^a}\,.
\end{equation}
It follows that $\psi'[\vv{m+t_i^a,m}]=\psi'[\vv{m,m+t_i^a}]$. 
\end{proof}
Each coefficient $c(m)$ can be calculated by following the path to $m$ from some fixed reachable $m_0$, whose coefficient can be chosen arbitrarily. Setting $c(m_0)=1$,
\begin{equation}
  c(m) = \rme^{ j[p(m_0\rightarrow m)]/2}= \Bigl\{ \prod_{e\in p(m_0\rightarrow m)} \frac{\nu(m'(e),m{''}(e))}{\nu(m{''}(e),m'(e))} \Bigr\}^{1/2}\,,
\end{equation}
where $p(m_0\rightarrow m)$ is any path, connecting $m_0$ with $m$.

Since $\partial_\tau{\sf D}^{-1}f=\aherm{\sf D}^{-1}f$ where $\aherm$ is anti-Hermitian, and $({\sf D}^{-1}f)_i=c(m)f(m)$ for some $m=m(i)$ by \theref{the:transformation}, we conclude that the evolution operator \eref{e:green} is unitary, provided that the inner product on the vector space $\Hilb_N$ is defined as follows
\begin{equation}\label{e:metric}
  \avg{f^1,f^2}_{\Hilb_N} = \sum_{n=1}^N\sum_{m\in\nodes_n} c^2(m) f^1_n(m) f^2_n(m) \quad f^1, f^2\in\Hilb_N.
\end{equation}

\section{Invariant subspace decomposition}\label{s:symmetries}
The following definition of an induced graph will be useful for the discussion.
\begin{definition} A graph $\graph^1=(\nodes^1,\edges^1)=\mathrm{ind}(\graph,m_1)$ is a subgraph of $\graph=(\nodes,\edges)$ induced by $m_1\in\nodes$, if $\nodes^1$ and $\edges^1$ are defined as follows. The set $\nodes^1\subset\nodes$ consists of all the nodes $m\in\nodes$ such that $m_1$ is reachable from $m$, and $\edges^1\subset\edges$ is the set of all the edges $\vv{m,m_1} \forall m$, by which $m_1$ can be reached from $m$.
\end{definition}
A graph $\graph$ is \emph{fully connected} if all of its nodes are reachable from any one of them. The same can be stated using induced graphs as $\mathrm{ind}(\graph,m)=\graph$ for any $m\in\nodes$. Conversely, if for some $m_1\in\nodes$, $\mathrm{ind}(\graph,m_1)\neq\graph$, then $\graph$ is a \emph{disconnected graph}. A disconnected graph is a union of several connected graphs $\graph^i$: $\graph=\graph^1\cup\graph^2 \cup \ldots$

If the {BBGKY} hierarchy \eref{e:bbgky3} is represented by a disconnected graph $\graph$, then all expansion coefficients in $\Hilb_N$ can be grouped by the node set $\nodes^i$ of its each connected subgraph $\graph^i=(\nodes^i,\edges^i)$. Each group of coefficients evolves by a subset of {BBGKY} hierarchy equations for these coefficients. Since different subgraphs $\graph^i$ are disconnected, the {BBGKY} hierarchy equations do not couple coefficients in different $\nodes^i$s and therefore each group of coefficients evolves independently from the remaining coefficients. Thus each set of nodes $\nodes^i$ defines indices of a group of expansion coefficients that belong to an invariant subspace $\Hilb_N^i$ of the full vector space $\Hilb_N$. In this section it will be shown that the connectivity of $\graph$ depends on the structure of the effective couplings, and therefore $\oJ$ determines the topology of $\Hilb_N$.

\subsection{Effective coupling constants}\label{s:oJ}
In \cite{paskauskas2012a} and in \sref{s:summary} it was shown that the effects of magnetic field $h$ can be accounted for, in thermodynamic limit, by replacing the matrix of coupling constants $J$ by a matrix of effective coupling constants $\oJ=(\oJ_{ab})$ which is defined by \eref{e:Javg}. If $h=0$, the interaction picture is not necessary since the equations with a non-averaged $J$ are of the same form \eref{e:bbgky2}. In this case assume $\oJ=J$. A common, but not the most general setting is when the couplings $J$ are diagonal, 
\begin{equation}\label{e:Jdiag}
  J=\mathrm{diag}(J_x,J_y,J_z)
\end{equation}
and $h=(h_x,h_y,h_z)$ is an arbitrary vector with a condition $h\neq 0$. Define $\hat{h}=h/\abs{h}$. In this case, the matrix of effective coupling constants $\oJ$ can be written as
\numparts\begin{equation}\label{e:Jbar}
  \oJ = \frac{J_d}{2} - \frac{J_s}{2}\,,
\end{equation}
where
\begin{eqnarray}
  J_d =\left\{\hat{h}_x^2(J_y+J_z) + \hat{h}_y^2(J_x+J_z) + \hat{h}_z^2(J_x+J_y)\right\}\one_{3\times 3} \label{e:Jbard}\\
  J_s  = \left\{J_x+J_y+J_z-3(\hat{h}_x^2J_x+\hat{h}_y^2J_y+\hat{h}_z^2J_z)\right\} (\hat{h}_a\otimes\hat{h}_b)\label{e:Jbars}
\end{eqnarray}\endnumparts
and where $(\hat{h}_a\otimes\hat{h}_b)$ denotes the outer product of $\hat{h}=(\hat{h}_x,\hat{h}_y,\hat{h}_z)$. 

Inspecting \eref{e:Jbar}--\eref{e:Jbars} we conclude that if $J$ is given by \eref{e:Jdiag} and either $h=0$ or $h$ has only one non-zero element [for example $h=(0,0,h_z)$], then $J_s=0_{3\times 3}$ and $\oJ$ is a diagonal matrix. In general, if two components of $h$ are non-zero, as in $h=(h_x,0,h_z)$, then $\oJ$ has one non-zero off-diagonal element; if all three components of $h$ are non-zero, then all off-diagonal elements of $\oJ$ are non-zero in general. The first two cases will be discussed in detail in the following, while it is assumed that $\graph$ is fully connected if all the diagonal elements of $\oJ$ are unequal, and all off-diagonal elements are non-zero.

\subsection{Diagonal couplings}\label{s:diagonal}
In this section, a special case will be considered where the coupling matrix $\oJ$ is diagonal
\begin{equation}\label{e:J1}
  \oJ = \mathrm{diag}( \oJ_x, \oJ_y, \oJ_z )\,.
\end{equation}
Define $\oJ$ to be \emph{anisotropic} if all three coupling constants $\oJ_a$ are distinct. Equivalently $\oJ$ is anisotropic if $\oK_x=\oJ_y-\oJ_z\neq0$, $\oK_y=\oJ_z-\oJ_x\neq0$, and $\oK_z=\oJ_x-\oJ_y\neq 0$. A similar setup has been studied in \cite{paskauskas2012a} where the coupling matrix was diagonal but not anisotropic. This special case will be further studied in \sref{s:case}.

\Eref{e:J1} implies $\oW_x=\oW_y=\oW_z=0$. As a result, $v_{n}^\pm(m,f)$ simplify considerably
\numparts
\begin{eqnarray}
  \fl v_{n}^-(m;f) &=& \oK_xm_ym_z f(m_x+1,m_y-1,m_z-1) + \oK_ym_zm_x f(m_x-1,m_y+1,m_z-1) \nonumber \\
  \fl  && +  \oK_zm_xm_y f(m_x-1,m_y-1,m_z+1)\,.\label{e:v-2} \\
  \fl v_{n}^+(m;f) &=& - \oK_xm_x f(m_x-1,m_y+1,m_z+1) -\oK_ym_y f(m_x+1,m_y-1,m_z+1)\nonumber\\
  \fl && - \oK_zm_z f(m_x+1,m_y+1,m_z-1)\,.\label{e:v+2}
\end{eqnarray}\endnumparts
It follows from \eref{e:v-2}--\eref{e:v+2} that all the edges of $\graph$ are generated by $t_1^a$, defined by \eref{e:tai}. 

Assume that $\oJ$ is anisotropic. Define
\numparts
\begin{eqnarray}
  \modes_1^1=\{ (1,0,0)\}, \quad 
  \modes_1^2=\{ (0,1,0)\}, \quad 
  \modes_1^3=\{ (0,0,1)\}, \\
  \modes_2^4=\{ (2,0,0), (0,2,0), (0,0,2)\}.
\end{eqnarray}\endnumparts
Define four subgraphs of $\graph$, induced by these sets
\numparts\begin{eqnarray}
\graph^1=(\modes^1,\edges^1)=\mathrm{ind}(\graph,\modes_1^1),\\
\graph^2=(\modes^2,\edges^2)=\mathrm{ind}(\graph,\modes_1^2),\\
\graph^3=(\modes^3,\edges^3)=\mathrm{ind}(\graph,\modes_1^3),\\
\graph^4=(\modes^4,\edges^4)=\mathrm{ind}(\graph,\modes_2^4).
\end{eqnarray}\endnumparts
Here $\modes_n^i\subset\nodes_n$. Similarly to the composition of $\nodes^{(N)}$ \eref{e:setM}, we define
\begin{equation}
  \modes^{(N,i)} = \bigcup_{n=1}^N \modes_n^i\,.
\end{equation}
Set $\modes_1^4$ is empty. The following lemma describes the structure of the remaining $\modes_n^i$.
\begin{lemma}\label{l:2step}For any $1\le n<N-2$, and $i=1,2,3,4$,
  \begin{equation}\label{e:iter2}
    \modes^i_{n+2}= \bigcup_{m\in\modes^i_n} \{ m + (2,0,0), m + (0,2,0), m + (0,0,2) \}
  \end{equation}
\end{lemma}
\begin{proof}
  Assume that $n$ is fixed and that \eref{e:iter2} is valid for $n-2$ and $n$. Induction of nodes in one step from $\modes_n^i$ to $\modes_{n+1}^i$ is obtained by following all the edges which are defined by vectors $t_1^a$: $m\to m+t_1^a$. Induction in two steps results in
  \begin{equation}
    \modes^i_{n+2}=\bigcup_{\scriptstyle ab\in\setI\atop \scriptstyle m\in\modes^i_n} \{ m + \{t_1^a+t_1^b\}\}\,.
  \end{equation}
where
\begin{equation}\label{e:iter2a}
  \fl\{ t_1^a+t_1^b\} =\{ (0,0,2), (0,2,0), (2,0,0), (-2,2,2), (2,-2,2), (2,2,-2) \}
  \end{equation}
  The first three terms in \eref{e:iter2a} satisfy \eref{e:iter2}. Each of the remaining three terms can be viewed as a three step path, e.g. $(-2,2,2)= (-2,0,0) + (0,2,0) + (0,0,2)$. By assumption $m'=m+(-2,0,0)\in \modes_{n-2}^i$ and $m''=m'+(0,2,0)\in\modes_n^i$. The last step $m'''=m''+(0,0,2)$ satisfies \eref{e:iter2} therefore $m'''\in\modes^i_{n+2}$.

By direct calculation one verifies that $\modes^1_2=\{(0,1,1)\}$, $\modes^1_3=\{(1,0,2), (1,2,0)\}$, and $\modes^1_4=\{(2,1,1),(0,1,3), (0,3,1)\}$ satisfy the assumption of the lemma for $n=1,3$, and $n=2,4$. The corresponding sets with $i=2,3$ are obtained by permutations of $\modes_n^1$.
The case of $i=4$ also satisfies the assumption with $n=2,4$ and $n=3,5$ as $\modes^4_3=\{(1,1,1)\}$, $\modes^4_4=\{(0,2,2),(2,0,2),(2,2,0)\}$, and $\modes_5^4=\{(3,1,1), (1,3,1), (1,1,3)\}$.
\end{proof}
The main result of this section is contained in the following theorem. 
\begin{theorem}\label{the:graph}
Assume that $\oJ$ is a diagonal anisotropic matrix. Then $\graph^i$ are disjoint subgraphs of $\graph$, and $\graph=\graph^1\cup\graph^2\cup\graph^3\cup\graph^4$. 
\end{theorem}
\begin{proof}
For the first part of the theorem it is sufficient to show that for each $1\le n\le N$,
\begin{equation}\label{e:cap}
  \modes_n^i\cap\modes_n^j=0\qquad i\neq j, \forall n.
\end{equation}
For the second part one must show that
\begin{equation}\label{e:cup}
  \nodes_n = \bigcup_{i=1}^4\modes_n^i  \qquad \forall n.
\end{equation}
In analogy with \eref{e:dn}, denote the number of elements in $\modes_n^i$ by $d_i(n)$. The connectivity structure of $\graph^1$, $\graph^2$ and $\graph^3$ is the same because their node sets $\modes^i$ are related by cyclic permutations, in addition to $\oK_x\neq0$, $\oK_y\neq0$, $\oK_z\neq0$, implied by the anisotropy of $\oJ$. Therefore the corresponding node sets $\modes_n^i$ have the same number of elements for each $n$, and the condition \eref{e:cup} can be replaced by \eref{e:cap} together with
\begin{equation}\label{e:proofdn}
  d(n) = 3d_1(n) + d_4(n) \qquad  \forall n.
\end{equation}
For $n=1,2$ conditions \eref{e:cap} and \eref{e:cup} are verified directly. In particular
\numparts\begin{eqnarray}
  \nodes_1=\{(1,0,0),(0,1,0),(0,0,1)\}=\cup_{i=1}^4\modes_1^i\,,\\
  \nodes_2=\{(2,0,0),(0,2,0),(1,1,0),(1,0,1),(0,1,1)\}=\cup_{i=1}^4\modes_2^i\,.
\end{eqnarray}\endnumparts
Consider the parity of  $m^x$, $m^y$, $m^z$ of the elements $m=(m^x,m^y,m^z)\in\modes_n^i$. Induction of nodes in one step by $m\to m+t_1^a$, adds or subtracts unity to each element $m^a$ of $m$. This operation inverts the parity of each $m^a$ in one step. In two steps, induction does not alter parity, as can be seen from \lemref{l:2step}. By examining $\modes_1^i$, $\modes_2^i$ and applying this simple parity alternation rule, the parity composition of $\modes_n^i$ can be determined. Define $i_a$ by $i_x=1$, $i_y=2$, $i_z=3$. For all $k=0,1,\dots, N/2$ and $a=x,y,z$
\numparts
\begin{eqnarray}
  m\in\modes_{2k}^{i_a}   &:& m^a\, \mathrm{is\,\, even,}\quad m^{\pi a}, m^{\pi^2a}\, \mathrm{are\,\, both\,\, odd},\label{e:parity1}\\
  m\in\modes_{2k+1}^{i_a} &:& m^a\, \mathrm{is\,\, odd,} \quad m^{\pi a}, m^{\pi^2a}\, \mathrm{are\,\, both\,\, even}.\label{e:parity2}
  \end{eqnarray}
Similarly, if $i=4$, 
\begin{eqnarray}
  m\in\modes_{2k}^4 &:& m^x, m^y, m^z\, \mathrm{are\,\, all\,\, even},\label{e:parity3}\\
  m\in\modes_{2k+1}^4  &:& m^x,m^y,m^z\, \mathrm{are\,\, all\,\, odd}.\label{e:parity4}
\end{eqnarray}\endnumparts
It follows from \eref{e:parity1}--\eref{e:parity4} that the patterns of parity are incompatible among all $\modes_n^i$ for each $n$, and therefore these sets are disjoint. This proves \eref{e:cap}.

Counting of the elements in $\modes_n^1$ and $\modes_n^4$ can be achieved by parameterizing all $m=(m^x,m^y,m^z)$ with the necessary properties\eref{e:parity1}--\eref{e:parity4}. Define 
\begin{equation}\label{e:mij}
  m_{kij}=(2k-2i,2i-2j,2j)\qquad k\ge i\ge j\ge 0.
\end{equation}
Using this definition all the cases in \eref{e:parity1}--\eref{e:parity4} can be parameterized as follows,
\numparts\begin{eqnarray}
  \modes_{2k+1}^1 &= \{ m= m_{kij} + (1,0,0) \} \qquad k\ge0\,, \label{e:d1}\\
  \modes_{2k+2}^1 &= \{ m=m_{kij} + (0,1,1) \}  \qquad k\ge0\,, \\
  \modes_{2k}^4 &= \{ m=m_{kij} \} \qquad k\ge1\,, \\
  \modes_{2k+3}^4 &= \{ m=m_{kij}+(1,1,1) \} \qquad k\ge0\,.\label{e:d4}
\end{eqnarray}\endnumparts
Parametrizations of $\modes^2_n, \modes^3_n$ are derived similarly from $\modes_n^1$.  From \eref{e:d1}--\eref{e:d4} it follows that the number elements in various sets $\modes^i_n$ is related to the number of elements $m_{ijk}$ with a fixed $k$ by
\begin{eqnarray}
  d_1(2k+1)=d_1(2k+2)=d_4(2k)=d_4(2k+3) \nonumber \\
  =\sum_{i,j}\big(1_{k\ge i\ge j}\big)\big(1_{i \ge j\ge 0}\big)=\sum_{i=0}^k (i+1) = d(k)\,.
\end{eqnarray}
where $1_X=1$ if $X$ is true and $1_X=0$ if $X$ is false. Counting of terms by \eref{e:d1}--\eref{e:d4}, done separately for the even $n=2k$ and odd $n=2k+1$ results in, respectively
\numparts\begin{eqnarray}
  3d_1(2k)+d_4(2k)= 3d(k-1)+d(k) = d(2k)\,,\\
  3d_1(2k+1)+d_4(2k+1)=3d(k)+d(k-1) = d(2k+1)\,.
\end{eqnarray}\endnumparts
Since the argument on the right hand side is $n$ in both cases, this proves \eref{e:proofdn}. 
\end{proof}
 
\subsection{Non-diagonal couplings}\label{s:offdiag}
Assume that one off-diagonal element of the coupling matrix is non-zero, i.e.  $\oW_a\neq0$ for some $a\in\{x,y,z\}$, and $\oW_{\pi a}=\oW_{\pi^2a}=0$. For example if $a=x$, $\oJ$ has the form
\begin{equation}\label{e:Joffdiagonal}
  \oJ = \left(\begin{array}{ccc} \oJ_x & \oW_x & 0\\ \oW_x & \oJ_y & 0\\ 0 & 0 & \oJ_z\end{array}\right)
\end{equation}
Define
\begin{equation}
  \modes^A= \modes^{i_a} \cup\, \modes^4\,, \qquad 
  \modes^B= \modes^{i_{\pi a}}\cup \modes^{i_{\pi^2 a}}\,.
\end{equation}
Define graphs $\graph^A\subset\graph$ and $\graph^B\subset\graph$ as induced from these sets
\numparts\begin{eqnarray}
  \graph^A=(\modes^A, \edges^A)=\mathrm{ind}(\graph,\modes^A)\,,\\
  \graph^B=(\modes^A, \edges^A)=\mathrm{ind}(\graph,\modes^B)\,.
\end{eqnarray}\endnumparts
\begin{theorem}\label{the:graph2} Assume that $\oK_x\neq 0$, $\oK_y\neq 0$, $\oK_z\neq0$, and $\oW_a\neq0$ for some $a\in\{x,y,z\}$ but $\oW_{\pi a}=\oW_{\pi^2a}=0$. Then graphs $\graph^A$ and $\graph^B$ are disjoint subgraphs of $\graph$, and $\graph=\graph^A\cup\graph^B$.
\end{theorem}
\begin{proof}
  Here, the complete node sets $\modes^A$ and $\modes^B$ are provided and it is sufficient to show that induction of nodes, generated by vectors $t_1^x$, $t_1^y$, $t_1^z$, and $t_2^a$, $t_3^a$, $t_4^a$ which correspond to non-zero effective coupling coefficients, leaves the node sets invariant. The first three vectors generate the sets $\modes^{1/2/3/4}$, discussed extensively in \sref{s:diagonal}. From the definition \eref{e:tai} it follows that induction in one step by the remaining three vectors inverts the parity of $m^a$ only, and does not change the parity of $m^{\pi a}$ and $m^{\pi^2a}$. By inspecting the parity structure of $\modes^i_n$, discussed in \theref{the:graph}, it follows that the addition of these vectors modifies the adjacency relations to $\adj^\pm(\modes^{i_a}_n)=\{ \modes^{i_a}_{n\pm1},\modes^4_{n\pm1}\}$,  $\adj^\pm(\modes^4_n)=\{\modes^4_{n\pm1}, \modes^{i_a}_{n\pm1}\}$, and similarly, $\adj^\pm(\{\modes^{i_{\pi a}}_n,\modes^{i_{\pi^2a}}_n\})=\{\modes^{i_{\pi a}}_{n\pm1},\modes^{i_{\pi^2}}_{n\pm1}\}$, from which the invariance requirement of $\modes^A$ and $\modes^B$ follows. Furthermore, disjointness and completeness of the graphs $\graph^A$ and $\graph^B$ follows from from the corresponding properties of $\graph^i$ and $\oplus_i\graph^i$, proved in \theref{the:graph}.
\end{proof}

\subsection{Isotropic diagonal couplings}\label{s:case}
A rather different situation is encountered when $\oJ$ is diagonal and isotropic, as in
\begin{equation}\label{e:Jiso}
  \oJ=\mathrm{diag}( \oJ_\perp,\oJ_\perp,\oJ_z )\,.
\end{equation}
Now we have $\oW_x=\oW_y=\oW_z=\oK_z=0$. Since $\sum_i\oK_i=0$, $\oK$ must be of the form 
\begin{equation}
  \oK = (-\kappa,\kappa,0)\,,
\end{equation}
where $\kappa=\oJ_z-\oJ_\perp$. Since $\kappa$ is the only coupling coefficient appearing in each term on the right-hand side of \eref{e:graph}, it can be scaled to $\pm1$ by an additional transformation of the time variable. Since the evolution is unitary, we may set $\kappa=+1$, and thus
\begin{eqnarray}\label{e:gspec}
  \fl \partial_\tau f_n(m) = m^xf_{n+1}(m^x-1,m^y+1,m^z+1)- m^yf_{n+1}(m^x+1,m^y-1,m^z+1) \nonumber\\
 \fl  + m^z\big[m^x f_{n-1}(m^x-1,m^y+1,m^z-1) - m^yf_{n-1}(m^x+1,m^y-1,m^z-1)\big]\,.
\end{eqnarray}
Since each term on the right-hand side of \eref{e:gspec} is multiplied by $m^x$ or $m^y$, the left-hand size of \eref{e:gspec} is zero whenever $m$ is of the form $m_n=(0,0,n)$. As a result, $f_n(m_n)$ is a constant of motion for each $n=1,2,\dots, N$.

To determine the remaining invariant subspaces, it is instructive to study a graph that corresponds to $\partial_\tau^2 f = \liouv^2 f$. A straightforward differentiation of \eref{e:gspec} results in
\begin{eqnarray}
  \fl\partial_\tau^2 f_n(m) =&-(2m^z+1)\bigl\{C_0f_n(m)-C_xf_n(m+\tpe)-C_yf_n(m-\tpe)\bigr\} \nonumber\\
  \fl &+C_x\bigl\{f_{n+2}(m+\tpe+\tpi)+C_zf_{n-2}(m+\tpe-\tpi)\bigr\}\nonumber\\
  \fl &+C_y\bigl\{f_{n+2}(m-\tpe+\tpi)+C_zf_{n-2}(m-\tpe-\tpi) \bigr\} \nonumber\\
  \fl &-C_0\bigl\{f_{n+2}(m+\tpi)+C_zf_{n-2}(m-\tpi)\bigr\}\,. \label{e:gspec2}
\end{eqnarray}
where 
\numparts\begin{equation}
  \fl C_a(m)=m^a(m^a-1)\quad a\in\{x,y,z\},\qquad C_0(m)=2m^xm^y+m^x+m^y\,,
\end{equation}
and where
\begin{equation}\label{e:tpie}
  \tpe=(-2,2,0) \qquad \tpi=(0,0,2).
\end{equation}\endnumparts
As in previous sections, we associate a graph $\graph=(\nodes,\edges)$ to the operator $\liouv^2$ on the right-hand side of \eref{e:gspec2}, and determine the disconnected subgraphs of $\graph$ by following all possible paths, induced by \eref{e:tpie} from various initial elements $m$. For example, all changes in $m^z$ are provided by $\pm\tpe$, which implies that the node set of any disconnected subgraph $\graph^i\subset\graph$ is such that $m^z$ is either an even or an odd integer from an interval $[m^z_{\mathrm{min}},N]$, where $0\le m^z_{\mathrm{min}}<N$ is the smallest possible value of $m^z$ of a particular subgraph. Moreover, if $\graph^i$ is a disconnected subgraph with $m^z=0,2,\dots$, then another graph $\graph'^{i}$, which differs from $\graph^i$ only by replacing $m^z$ by $m^z+1$, is also a disconnected subgraph of $\graph$: $\graph'^{i}\subset\graph$.

The ``horizontal'' induction of nodes (by paths that do not change $m^z$, nor $\abs{m}$) is obtained by applying $m\mapsto m\pm\tpe$ iteratively with the usual constraint $m\in\nodes$. Four cases are possible, depending on the parity of $m^x$ and $m^y$. Define $\kodes_{1,kn}, \kodes_{2,kn}\subseteq\nodes_{n+2k+1}$, $\kodes_{3,kn}\subseteq\nodes_{n+2k}$, and $\kodes_{4,kn}\subseteq\nodes_{n+2k+3}$ by
\numparts\begin{eqnarray}
\fl\kodes_{1kn}=\bigcup_{j=0}^k (1+2j,2k-2j,n)\,, \quad & \kodes_{2kn}=\bigcup_{j=0}^k (2k-2j,1+2j,n) \,, \label{e:kodes-1}\\
\fl\kodes_{3kn}=\bigcup_{j=0}^k (2k-2j,2j,n) \,,        & \kodes_{4kn}=\bigcup_{j=0}^k (1+2k-2j,1+2j,n) \,.\label{e:kodes-2}
\end{eqnarray}\endnumparts
Each of the sets $\kodes_{ikn}\subseteq\nodes_{n'}$ is invariant with respect to the horizontal induction. 

Define even and odd node sets by combining the corresponding $\kodes_{ikn}$ as follows
\begin{equation}
  \Kodes^{e}_{ik} = \bigcup_{n=0}^{N/2} \kodes_{ik,2n} \qquad \Kodes^{o}_{ik} = \bigcup_{n=0}^{N/2} \kodes^3_{k,2n+1}\,.
\end{equation}
By construction, $\Kodes^{e/o}_{ik}$ are invariant with respect to node induction by \eref{e:tpie}, therefore a family of induced graphs, parametrized by $e/o$, $i=1,2,3,4$ and $k=0,1,\dots, N/2$ and defined by
\begin{equation}\label{e:ON}
  \fl\graph^{e/o}_{ik} = (\Kodes^{e/o}_{ik}, \edges^{e/o}_{ik}) = \mathrm{ind}(\graph, \Kodes^{e/o}_{ik}) \quad k=0,1,\dots, N/2, \quad i=1,2,3,4
\end{equation}
are disconnected subgraphs of $\graph$. 

The range of the parameter $k$ in \eref{e:ON} implies that the number of disconnected subgraphs of $\graph$ scales as $O(N)$ and consequently, the same scaling holds for the number of invariant vector spaces $\Hilb^{e/o}_{ik,N}$, associated to each disconnected subgraph $\graph^{e/o}_{ik}$.

Associated to each $\kodes_{ikn}$ there is an $n'$-spin vector space $\hilb_{n'}$, where $n'=n+2k+1$ if $i=1,2$, $n'=n+2k$ if $i=3$, and $n'=n+2k+3$ if $i=4$. By definition \eref{e:kodes-1}--\eref{e:kodes-2}, $\kodes_{ikn}$ has $k+1$ elements, thus each $\hilb_{n'}$ subspace is $k+1$ dimensional. As a result, the operator $\liouv^2$ associated to a disconnected subgraph $\graph^{e/o}_{ik}$ has a block structure where all the blocks are square $(k+1)\times (k+1)$ size matrices, and where $k$ is fixed.

\subsection{Dependence on initial conditions}
In the preceding sections it was demonstrated that, depending on the structure of $\oJ$, connectivity properties of $\graph$ vary. They may be summarized as follows.
\begin{equation}
  \graph = \bigcup_{\alpha_i\in\alpha} \graph^{\alpha_i}\,,
\end{equation}
where $\alpha=\{\alpha_1,\alpha_2,\dots \}$ is a set of partition indices and where $\graph^{\alpha_i}=(\nodes^{\alpha_i},\edges^{\alpha_i})$ is a connected graph for each $\alpha_i\in\alpha$. If $\oJ$ is a diagonal anisotropic matrix as defined in \sref{s:diagonal}, $\alpha=\{1,2,3,4\}$; if $\oJ$ is diagonal anisotropic and, in addition, one off-diagonal element of $\oJ$ is non-zero, as defined in \sref{s:offdiag}, $\alpha=\{A,B\}$. Lastly, if $\oJ$ is a diagonal isotropic matrix then, as discussed in \sref{s:case}, $\alpha=\cup_k\cup_j (e/o,j,k)$.

A collection of components $f(m)$, $m\in\nodes^{\alpha_i}$ was viewed as a vector in a vector space $\Hilb^{\alpha_i}_N\subseteq\Hilb_N$. Connectendess of a graph $\graph^{\alpha_i}$ implies that the subset of {BBGKY} equations for the evolution of vectors in $\Hilb^{\alpha_i}_N$ involves only vectors in $\Hilb^{\alpha_i}_N$ and therefore $\Hilb^{\alpha_i}_N$ is an invariant subspace of $\Hilb_N$. It follows that
\begin{equation}\label{e:liouvdecomp}
  \liouv = \bigoplus_{\alpha_i\in\alpha}\liouv^{\alpha_i}\,.
\end{equation}
where $\liouv$ is the generalized Liouville operator such that $\liouv^{\alpha_i}: \Hilb^{\alpha_i}_N\mapsto\Hilb^{\alpha_i}_N$ acts on each invariant subspace.

Almost every $\graph^{\alpha_i}$ (the only exception is $\alpha=(e/o,3,k)$ in \sref{s:case}) slices across the entire graph $\graph$, meaning that $\graph^{\alpha_i}$ describes tensor coefficients from $n$-spin vector subspace with all $n_{\mathrm{min}}\le n\le N$ (possibly $n$ is constrained to either even or odd values) where $n_{\mathrm{min}}$ is some smallest order spin vector space $\hilb_{n_{\mathrm{min}}}$ of that subspace $\Hilb_N^{\alpha_i}$. As a result, the initial value problem \eref{e:iv} in each $\Hilb^{\alpha_i}_N$ has to be solved to all orders in the expansion coefficients $f_n$ in general.

From \eref{e:liouvdecomp} it follows that dynamics a-priori depends on the initial conditions. For example, if the initial condition is such that the only non-zero components of $f$ are in one of the invariant subspaces $\Hilb^{\alpha_i}_N$, then the properties of equilibration and the equilibrium state will be determined by, respectively, the spectrum and the null-space of $\liouv^{\alpha_i}$. Both these properties may depend on $\alpha_i$ and thus on the initial state.

\section{Conclusions}\label{s:conclusions}
Aspects of the dynamics of the anisotropic Curie-Weiss quantum Heisenberg model are addressed by studying the structure of its equations of motion. This study is based on the {BBGKY} theory and a special expansion of reduced density operators in a basis of Pauli operators, and it builds on a recent work where the coefficient {BBGKY} hierarchy \eref{e:bbgky2} was derived \cite{paskauskas2012a}.

This work was motivated by a special case of coupling constants \cite{paskauskas2012a}, where time dependence of a certain subset of spin-spin correlators could be determined analytically to all orders, allowing to determine the evolution of a class of observables. It was hoped that the method of solution could be extended to all observables, initial conditions, and more general coupling constants. As in became apparent that a special symmetry of the coefficient vector space was behind the analytic solution, a comprehensive study of symmetries of a general case was called for.

To examine the structure of the coefficient vector space, the {BBGKY} hierarchy was represented by a graph, whose vertices are the expansion coefficients, and edges are couplings in the {BBGKY} hierarchy. It was shown that, under certain conditions, this graph is disconnected. A disconnected graph implies that certain groups of expansion coefficients are causally unrelated. To each such group one can associate an invariant subspace of the underlying vector space. It therefore follows that the vector space is partitioned into invariant subspaces, and that this partition depends on the structure of the matrix of effective coupling constants. Three scenarios of partition were identified in addition to the most general case: If $\oJ$ is a diagonal anisotropic matrix, there are four invariant subspaces. If in addition $\oJ$ has one non-zero off-diagonal element, there are two invariant subspaces. In the most general case, the vector space is irreducible, which corresponds to $\oJ$ that is diagonally anisotropic and has two or three non-zero off-diagonal elements. Lastly, it was found that, if $\oJ$ is diagonal and isotropic, the number of invariant subspaces scales as $O(N)$ with the size of the system $N$.

In addition it was shown that, to make the evolution in the coefficient vector space unitary, it is necessary to scale the components of vectors by certain weight factors, which is equivalent to redefining the conventional inner product, where the products of vector components are multiplied by the squares of these factors.

This work provides several insights about symmetries of the Hilbert space of a simple quantum model. It is believed in this work that similar symmetries are present in a class of more realistic lattice spin models, where interactions among particles depend on the distance between them. It is also believed that these symmetries should be taken into consideration when low-dimensional approximations, such as kinetic theories, are constructed.

\appendix

\section{Derivation of~\eref{e:v-} and~\eref{e:v+}}\label{s:derivation}
Using permutation symmetry of the tensors $f_n^a$, components of $v_n^\pm$ defined by \eref{e:v-}--\eref{e:v+} can be written as
\numparts\begin{eqnarray}
(v_n^+)^{(a_1\dots a_n)} = \sum_{uvw\in \setI} \sum_{i=1}^n \eps^{a_iuv}(-J^{uw}) f_{n+1}^{wv \dots } \label{e:v+0} \\
(v_n^-)^{(a_1\dots a_n)} = \sum_{uv\in \setI} \sum_{\scriptstyle i,j=1\atop\scriptstyle i\neq j}^n \eps^{a_iuv}(-J^{a_ju}) f_{n-1}^{v \dots } \label{e:v-0}
\end{eqnarray}\endnumparts
where the dots on the right-hand side denote the multiindex $(a_1\dots a_n)$ modified by: deleting $a_i$ in \eref{e:v+0}, and deleting both  $a_i$ and $a_j$ in \eref{e:v-}.

For $a\in \{x,y,z\}$, define $b_a=\pi a$ and $c_a=\pi^2a$. Define the permutation operator $\pi_a$ acting on a triple $(a,b_a,c_a)$ by $\pi_a (a, b_a, c_a) = (x,y,z)$. Due to the permutation invariance of $f_n$, the sums over $i$ and $j$ in \eref{e:v+0}--\eref{e:v-0} can be replaced by sums over coordinate labels, multiplied by the number of occurrences of that label in the multiindex $(a_1\dots a_n)$. For example in the term \eref{e:v+0} this multiplicity factor is equal to $m^a$, by definition \eref{e:setMn}. Passing to a notation $f_n^{a(m)}\to f_n(m)$ in the term $(v_n^+)^{(a_1,\dots,a_n)}$ and performing $u$ and $w$ summation explicitly, this term is rewritten as
\begin{eqnarray}\label{e:tmp0}
  (v_n^+)^{(a_1\dots a_n)} = \sum_{auvw\in\setI} m^a \eps^{auv}(-J_{uw}) f_{n+1}^{wv\dots} \nonumber \\
  \fl  = \sum_{aw\in\setI} m^a \big\{ (-J_{b_aw})f_{n+1}[\pi_a (m^a-1,m^{b_a}+\delta(b_a,w),m^{c_a}+1+\delta(c_a,w))]\nonumber\\
  \fl + J_{c_aw}f_{n+1}[\pi_a(m^a-1+\delta(a,w),m^{b_a}+1+\delta(b_a,w),m_c+\delta(c_a,w))]\big\}
\end{eqnarray}
where $\delta(a,b)$ is the Kroenecker symbol, acting on coordinate labels. In arriving at this expression, note that the term $m^a-1$ arises because  $a_i$ is deleted from the expression $f_{n-1}^{wv\dots}$, while the terms $m^{b_a}+1$ and $m^{c_a}+1$ arise from the $v$-summation. 

Performing the $w$-summation, and passing to notation $(v_n^+)^{(a_1\dots a_n)}\to v_n^+(m)$, we find 
\begin{eqnarray}\label{e:tmp1}
  \fl v_n^+(m) = \sum_a m^a\Big\{ -(J^{b_ab_a}-J^{c_ac_a})f_{n+1}[\pi_a(m^a-1,m^{b_a}+1,m^{c_a}+1)] \nonumber \\
  \fl - J_{b_ac_a}\big( f_{n+1}[\pi_a(m^a-1,m^{b_a},m^{c_a}+2)]-f_{n+1}[\pi_a(m^a-1,m^{b_a}+2,m^{c_a})]\big)  \nonumber\\
   - J_{b_ac_a}\big( m^{b_a}-m^{c_a}\big) f_{n+1}[\pi_a(m^a+1,m^{b_a},m^{c_a})]\Big\}
\end{eqnarray}
To arrive at this expression from \eref{e:tmp0}, a permutation invariance property was used: all the labels $a$, $b_a$, $c_a$ can be cyclically permuted without changing the value of the sum, and moreover $f[\pi_a(a,b_a,c_a)]=f[\pi_{b_a}(b_a,c_a,a)]$. For example
\begin{equation*}
  \fl \sum_a m^aJ_{c_aa}f_{n+1}[\pi_a(m^a,m^{b_a}+1,m^{c_a})]=\sum_a m^{c_a}J_{b_ac_a}f_{n+1}[\pi_a(m^a+1,m^{b_a},m^{c_a})]\,.
\end{equation*}
The final expression \eref{e:vplus} is obtained from \eref{e:tmp1} by identifying
\begin{equation}\label{e:KW}
  K_a=J_{b_ab_a}-J_{c_ac_a} \qquad W_a=J_{b_ac_a}
\end{equation}
The term \eref{e:v-0} is treated similarly, i.e. the sums over $i$ and $j$ are replaced by  sums over $u,v\in\setI$. Since the labels $a_i$ and $a_j$ may be equal for $i\neq j$ there are two cases of multiplicity factors to consider, $m^{a_i}(m^{a_i}-1)$ if $a_i=a_j$, and $m^{a_i}m^{a_j}$ if $a_i\neq a_j$. Separating these two cases, $v_n^-$ is written as 
\begin{eqnarray}
  \fl (v_n^-)^{(a_1\dots a_n)}=\sum_{auv\in\setI} m^a(m^a-1) \eps_{auv}(-J_{ua}) f_{n-1}^{v\dots}+\sum_{\scriptstyle auvw\atop \scriptstyle w\neq a} m^am^w \eps_{auv}(-J_{uw}) f_{n-1}^{v\dots} \nonumber\\
  \fl  =\sum_{a\in \setI} m^a(m^a-1)\Big\{J_{c_aa} f_{n-1}[\pi_a(m^a-2,m^{b_a}+1,m^{c_a})] \nonumber \\
  - J_{b_aa} f_{n-1}[\pi_a(m^a-2,m^{b_a},m^{c_a}+1)]\Big\} \nonumber\\
  \fl +\sum_{\scriptstyle a\in\setI\atop\scriptstyle w\in\{b_a,c_a\}} m^am^w\Big\{ J_{c_aw}f_{n-1}[\pi_a(m^a-1,m^{b_a}+1-\delta(b_a,w),m^{c_a}-\delta(c_a,w))] \nonumber \\
    - J_{b_aw} f_{n-1}[\pi_a(m^a-1,m^{b_a}-\delta(b_a,w),m^{c_a}+1-\delta(c_a,w))]\Big\}
\end{eqnarray}
In the second sum, performing the $w$-sum and after relabelling terms in the remaining $a$-sum, one finds  after passing to notation $(v_n^-)^{a(m)}\to v_n^-(m)$
\begin{eqnarray}\label{e:tmp2}
  \fl v_n^-(m)=\sum_{a\in\setI} \Bigl\{ (J^{b_ab_a}-J^{c_ac_a}) m^{b_a}m^{c_a} f_{n-1}[\pi_a(m^a+1,m^{b_a}-1,m^{c_a}-1)] \nonumber\\
  \fl + m^a(m^a-1)\bigl\{ J_{ac_a}f_{n-1}[\pi_a(m^a-2,m^{b_a}+1,m^{c_a})]-J_{ab_a}f_{n-1}[\pi_a(m^a-2,m^{b_a},m^{c_a}+1)]\bigr\} \nonumber \\
  + J_{b_ac_a} m^a \left(m^{b_a}-m^{c_a}\right) f_{n-1}[\pi_a(m^a-1,m^{b_a},m^{c_a})]\Bigr\}
\end{eqnarray}
The final form \eref{e:vminus} is obtained from \eref{e:tmp2} by substituting \eref{e:KW}. 

\vspace{3mm}
\bibliographystyle{unsrt}
\bibliography{Graphs.bbl}
\end{document}